\documentclass{sig-alternate-05-2015}
\usepackage{graphicx}
\usepackage{balance}
\usepackage{algorithm}
\usepackage{algpseudocode}
\usepackage{url}
\usepackage{color}
\usepackage{subcaption}

\begin{document}

\newcommand{\bh}{\boldsymbol{h}}
\newcommand{\bv}{\boldsymbol{v}}
\newcommand{\bw}{\boldsymbol{w}}
\newcommand{\bx}{\boldsymbol{x}}
\newcommand{\by}{\boldsymbol{y}}
\newcommand{\bz}{\boldsymbol{z}}
\newcommand{\btheta}{\boldsymbol{\theta}}
\newcommand{\bphi}{\boldsymbol{\phi}}
\newcommand{\bpsi}{\boldsymbol{\psi}}
\newcommand{\bhaty}{\boldsymbol{\hat{y}}}
\newcommand{\argmin}{\mathop{\mathrm{argmin}}}
\renewcommand{\P}{\mathcal{P}}
\renewcommand{\O}{\mathcal{O}}
\newcommand{\separable}[1]{$#1$-separable}

\makeatletter
\def\@copyrightspace{\relax}
\makeatother

\title{A Generic Coordinate Descent Framework\\ for Learning from Implicit Feedback}

\numberofauthors{4}
\author{
\alignauthor
Immanuel Bayer\thanks{Work done at Google.}\\
\affaddr{University of Konstanz, Germany}\\
\email{immanuel.bayer@uni-konstanz.de}
\alignauthor
Xiangnan He\footnotemark[1]\\
\affaddr{National University of Singapore, Singapore}\\
\email{xiangnan@comp.nus.edu.sg}
\and
\alignauthor
Bhargav Kanagal\\
\affaddr{Google Inc., USA}\\
\email{bhargav@google.com}
\alignauthor
Steffen Rendle\\
\affaddr{Google Inc., USA}\\
\email{srendle@google.com}
}
\maketitle

\begin{abstract}
In recent years, interest in recommender research has shifted from explicit feedback towards implicit feedback data.
A diversity of complex models has been proposed for a wide variety of applications.
Despite this, learning from implicit feedback is still computationally challenging.
So far, most work relies on stochastic gradient descent (SGD) solvers which are easy to derive, but in practice challenging to apply, especially for tasks with many items.
For the simple matrix factorization model, an efficient coordinate descent (CD) solver has been previously proposed.
However, efficient CD approaches have not been derived for more complex models.

In this paper, we provide a new framework for deriving efficient CD algorithms for complex recommender models.
We identify and introduce the property of \separable{k} models.
We show that $k$-separability is a sufficient property to allow efficient optimization of implicit recommender problems with CD.
We illustrate this framework on a variety of state-of-the-art models including factorization machines and Tucker decomposition.
To summarize, our work provides the theory and building blocks to derive efficient implicit CD algorithms for complex recommender models.
\end{abstract}

\section{Introduction}

In recent years, the focus of recommender system research has shifted from explicit feedback problems such as rating prediction to implicit feedback problems.
Most of the signal that a user provides about her preferences is \emph{implicit}.
Examples for implicit feedback are: a user watches a video, clicks on a link, etc.
Implicit feedback data is much cheaper to obtain than explicit feedback, because it comes with no extra cost for the user and thus is available on a much larger scale.
However, learning a recommender system from implicit feedback is computationally expensive because the observed actions of a user need to be contrasted against {\em all} the non-observed actions \cite{Hu:icdm2008,rendle:uai09}.

Stochastic gradient descent (SGD) and coordinate descent (CD) are two widely used algorithms for large scale machine learning.
Both algorithms are considered state-of-the-art for learning matrix factorization models from implicit feedback and have been studied extensively.
SGD and CD have shown different strengths and weaknesses on various data sets  \cite{HeM16,ShiKaratzoglouBaltrunasEtAl2012,ShiKaratzoglouBaltrunasEtAl2012a,NingKarypis2011,Zhao14,sedhain2016effectiveness,volkovs2015effective,zhao2015improving}.
While SGD is available as a general framework to optimize a broad class of models~\cite{rendle:uai09}, CD is only available for a few simple models~\cite{Hu:icdm2008,PilaszyZibriczkyTikk2010}.
In fact, it is even unknown if CD can be used to efficiently optimize complex recommender models.
Our work closes this gap and identifies a model property called \emph{$k$-separability}, that is a sufficient condition to allow efficient learning from implicit feedback.
Based on $k$-separability, we provide a general framework to derive efficient implicit CD solvers.

Our paper is organized as follows:
First, we introduce the problem of learning from implicit feedback and show that the number of implicit training examples makes the application of standard algorithms challenging.
Next, we provide our general framework for efficient implicit learning with CD.
We identify $k$-separability of a model as a sufficient property to make efficient learning feasible and introduce iCD, a generic learning algorithm for $k$-separable models.
In Section~\ref{sec:applications}, we show how to apply iCD to a diverse set of models, including, matrix factorization (MF), factorization machines (FM) and tensor factorization.
This section serves both as solutions to popular models as well as a guide for applying the framework to other complex recommender models.

To summarize, our contributions are:
\begin{itemize}
  \item We identify a basic property of recommender models that allows efficient CD learning from implicit data.
  \item We provide iCD, a framework to derive efficient implicit CD algorithms.
  \item We apply the framework and derive algorithms for MF, MF with side information, FM, PARAFAC and Tucker Decomposition.
\end{itemize}

\section{Related Work}

Since several years, matrix factorization (MF) is regarded as the most effective, basic recommender system model.
Two optimization strategies dominate the research on MF from implicit feedback data.
The first one is Bayesian Personalized Ranking (BPR) \cite{rendle:uai09}, a stochastic gradient descent (SGD) framework, that contrasts pairs of consumed to non-consumed items.
The second one is coordinate descent (CD) also known as alternating least squares on an elementwise loss over both the consumed and non-consumed items \cite{Hu:icdm2008}.
In terms of the loss formulation, BPR's pairwise classification loss is better suited for ranking whereas CD loss is better suited for numerical data.
With regard to the optimization task, both techniques face the same challenge of learning over a very large number of training examples.
BPR tackles this issue by sampling negative items, but it has been shown that BPR has convergence problems when the number of items is large \cite{mcfee:www12,rendle:wsdm14}.
It requires more complex, non-uniform, sampling strategies for dealing with this problem \cite{rendle:wsdm14,kanagal:vldb12}. 
On the other hand, for CD-MF, Hu et al.~\cite{Hu:icdm2008} have derived an efficient algorithm that allows to optimize over the large number of non-consumed items without any cost.
This computational trick is exact and does not involve sampling.
Many authors have compared both CD-MF and BPR-MF on a variety of datasets and some work reports better quality for BPR-MF \cite{HeM16,ShiKaratzoglouBaltrunasEtAl2012,ShiKaratzoglouBaltrunasEtAl2012a,NingKarypis2011} whereas for other problems CD-MF works better \cite{NingKarypis2011,Zhao14,sedhain2016effectiveness,volkovs2015effective,zhao2015improving}.
This large body of results indicates that the advantages of CD and BPR are orthogonal and both approaches have their merits.

Our discussion so far was focused on learning matrix factorization models from implicit data.
Shifting from simple matrix factorization to more complex factorization models has shown large success in many implicit recommendation problems~\cite{GantnerDrumondFreudenthalerEtAl2010,HeM16,ShmueliKagianKorenEtAl2012,ChengYangLyuEtAl2013,PanChen2013,YuRenSunEtAl2014}.
However, work on complex factorization models relies almost exclusively on SGD optimization using the generic BPR framework.
Our work, provides the theory as well as a practical framework for deriving CD learners for such complex models.
Like CD for MF, our generic algorithm is able to optimize on all non-consumed items without explicitly iterating over them.
To summarize, our paper enables researchers and practitioners to apply CD in their work and gives them a choice between the advantages of BPR and CD.

\section{Problem Statement}

\begin{figure}[t]
  \centering
  \includegraphics[width=0.3\textwidth]{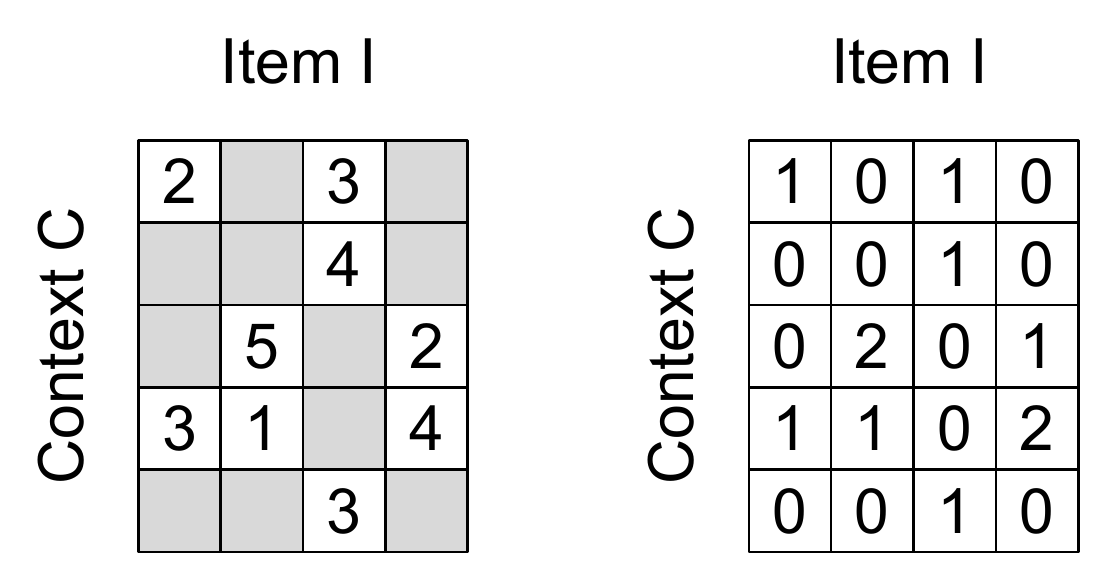}
  \caption{Left: Explicit rating data, with $S=\{(c_1,i_1,2),(c_1,i_3,3),(c_2, i_3,4),\ldots\}$. Right: Implicit data, e.g., watch/ purchase/ click count, $|S_\text{impl}| = |C||I|$.
  \label{fig:rec_data}}
\end{figure}

Let $I$ be a set of items and $C$ a set of contexts.
Let $S$ be a set of observed feedback where a tuple $(c,i,y,\alpha) \in S$ indicates that in context $c$, a score $y$ has been assigned to item $i$ with confidence $\alpha$.
See Figure~\ref{fig:rec_data} for an illustration.
We use a general notation of context which can include for instance user, time, location, attributes, history, etc.
Section~\ref{sec:applications} and Section~\ref{sec:experiments} show more examples for context.

\subsection{Recommender Model}
A recommender model $\hat{y} : C \times I \rightarrow \mathbb{R}$ is a function that assigns a score to every context-item pair.
The model $\hat{y}$ is parameterized by a set of model parameters $\Theta$.
The model $\hat{y}$ is typically used to decide which items to present in a given context.
\pagebreak

The learning task is to find the values of the model parameters that minimize a loss over the data $S$, e.g., a squared loss
\begin{align}
  L(\Theta|S) = \sum_{(c,i,y,\alpha) \in S} \alpha\,(\hat{y}(c,i) - y)^2 + \sum_{\theta \in \Theta} \lambda_\theta \theta^2 \label{eq:explicit_loss}
\end{align}
where $\lambda_\theta$ is an regularization constant for parameter $\theta$.

\subsection{Coordinate Descent Algorithm}
\label{sec:explicit_cd}

Objective (\ref{eq:explicit_loss}) can be minimized by coordinate descent (CD).
CD iterates through the model parameters and updates one parameter at a time. 
For a selected parameter $\theta \in \Theta$, CD computes the first $L'$ and second derivative $L''$ of $L$ with respect to the selected coordinate $\theta$:
\begin{align}
  L'(\theta|S) &= 2\,\sum_{(c,i,y,\alpha) \in S} \alpha\,(\hat{y}(c,i) - y) \hat{y}'(c,i) + 2\, \lambda_\theta\, \theta \\
  L''(\theta|S) &= 2\,\sum_{(c,i,y,\alpha) \in S} \alpha\,[(\hat{y}(c,i) - y) \hat{y}''(c,i) + \hat{y}'(c,i)^2] + 2\, \lambda_\theta
\end{align}
and performs a Newton update step:
\begin{align}
  \theta \leftarrow \theta - \eta\,\frac{L'(\theta|S)}{L''(\theta|S)} \label{eq:newton}
\end{align}
where $\eta \in (0,1]$ is the step size.
For multilinear models, a full step, i.e., $\eta=1$, can be chosen without risking divergence~\cite{Rendle:tist2012}.
All models in Section~\ref{sec:applications} fall into this category.

Such CD algorithms have been well studied and the runtime complexity is typically linear in the complexity of the training examples and embedding dimension.
For MF, \cite{yu:icdm12} shows a complexity of $\O(|S|\,k)$ and for FM, \cite{Rendle:tist2012} derives a complexity of $\O(N_Z(X)\,k)$ where $N_Z(X)$ is the number of non-zero entries in the design matrix $X$.
The linear runtime complexity in the number of training examples makes these algorithms well suited for explicit recommendation settings, however, they become infeasible for implicit problems.

\subsection{Learning from Implicit Feedback}
\label{sec:implicit}

In an implicit recommendation problem, the non-consumed items are meaningful and cannot be ignored.
For instance, in Figure~\ref{fig:rec_data} (right), the data depicts how often each item was consumed in a context in the past.
The non-consumed items, i.e., the ones with a count of zero, are useful to learn user preferences.
To formalize, the training data $S_{\text{impl}}$ of an implicit problem consists of a set $S^+$ of observed feedback and all the non-consumed tuples $S^0$
\begin{align}
  S_{\text{impl}} = S^+ \cup S^0, \quad |S_{\text{impl}}| = |C|\,|I| \label{eq:implicit_data}
\end{align}
with
\begin{align}
  \forall (c,i,y,\alpha) \in S^0 : y = 0, \alpha = \alpha_0 .
\end{align}
$S^+$ contains the observed feedback and is of much smaller scale than $S_\text{impl}$, usually $|S^+| \ll |C|\,|I|$. 

The implicit learning problem can be stated as minimizing the objective in eq.~(\ref{eq:explicit_loss}) over the implicit data $S_\text{impl}$.
While possible in theory, in practice, it is infeasible to apply the learning algorithms of Section~\ref{sec:explicit_cd} to this problem due to their linear computational runtime in the size of the training data which is $|S_\text{impl}| = |C||I|$ for implicit problems.
Our paper shows how to derive efficient CD algorithms for optimizing eq.~(\ref{eq:explicit_loss}) over implicit data.

\section{Generic Coordinate Descent Algorithm for Implicit Feedback}
\label{sec:method}

\subsection{Implicit Regularizer}

As discussed in Section~\ref{sec:implicit}, the reason why training on implicit data is challenging is the large number of implicit examples $S^0$ which is typically $|S^0| \in \O(|C||I|)$.
Note that $S^0$ includes all context-item pairs that are \textbf{not} in $S^+$.
We show now that we can rephrase the optimization criterion to sum over \textbf{all} context-item pairs.
This reformulation is a prerequisite to later allow the decomposition of the loss in Section~\ref{sec:decomposition}.
Moreover it allows to study implicit optimization without having to consider $S^+$.

\newtheorem{lemma}{Lemma}
\begin{lemma}
  \label{lemma:rephrase}
Implicit learning can be rephrased as a combination of learning on a small positive set and minimizing the scoring function on \textbf{any} context-item pair.  
\begin{align}
	\argmin_\Theta L(\Theta|S_\text{impl}) =  \argmin_\Theta \biggl( L(\Theta|S) + \alpha_0 \underbrace{\sum_{c \in C} \sum_{i \in I} \hat{y}(c,i)^2}_{ =: R(\Theta)} \biggr) \label{eq:iloss}
\end{align}
where the observed feedback is rescaled
\begin{align}
  S := \left\{\left(c,i,\frac{\alpha}{\alpha-\alpha_0}y,\alpha-\alpha_0\right) : (c,i,y,\alpha) \in S^+\right\} . \label{eq:rescale}
\end{align}
\end{lemma}
\begin{proof}
Per definition of the loss (eq.~\ref{eq:explicit_loss}) and the implicit training set $S_\text{impl}$ (eq.~\ref{eq:implicit_data})
\begin{align*}
  &L(\Theta|S_\text{impl}) = L(\Theta|S^+) + \alpha_0\sum_{(c,i) \in S^0} \hat{y}(c,i)^2 \\
 =& L(\Theta|S^+) - L(\Theta|\{(c,i,0,\alpha_0) : (c,i,y,\alpha) \in S^+\}) +  \alpha_0 R(\Theta)
\end{align*}
We further can collapse each pair of examples into a single one.
We show this for the pair $(c,i,y,\alpha) \in S$ and its counterpart $(c,i,0,-\alpha_0)$.
\begin{align*}
  &L(\Theta|\{(c,i,y,\alpha)\}) +  L(\Theta|\{(c,i,0,-\alpha_0)\}) \\
= &\alpha\,(\hat{y}(c,i) - y)^2 - \alpha_0\,\hat{y}(c,i)^2 \\
= &(\alpha - \alpha_0)\,\left(\hat{y}(c,i)^2 - 2\,\frac{\alpha}{\alpha - \alpha_0} y\,\hat{y}(c,i) + \frac{\alpha}{\alpha - \alpha_0} y^2\right) \\
= &(\alpha - \alpha_0)\left(\hat{y}(c,i) - \frac{\alpha}{\alpha - \alpha_0} y\right)^2 + \text{const} \\
= &L\left(\Theta|\left\{\left(c,i,\frac{\alpha}{\alpha - \alpha_0} y,\alpha-\alpha_0\right)\right\}\right) + \text{const}
\end{align*}
The additional constant does not change the optimum for $\Theta$, so rescaling of examples as in eq.~(\ref{eq:rescale}) preserves the optimum.
\end{proof}
The lemma allows an interesting interpretation of implicit learning tasks.
Implicit problems can be seen as explicit or one-class problems with an additional \emph{implicit regularizer} or bias $R(\Theta)$ for predicting zeros.
Compared to a common regularizer such as L2, the implicit regularizer is aware of the model $\hat{y}$.
L2 penalizes non-zero \emph{model parameters} $\Theta$ whereas the implicit regularizer penalizes non-zero \emph{predictions} $\hat{y}$.
Consequently, the implicit regularizer is less restrictive than L2 because small predictions can be achieved even with large model parameters.

\subsection{iCD Algorithm for \separable{k} Models}
\label{sec:decomposition}

As shown in eq.~(\ref{eq:iloss}), implicit learning can be formulated as explicit learning on a small set $S$ with an expensive implicit regularizer $R$.
Learning models over an explicit loss is already well studied~\cite{yu:icdm12,Rendle:tist2012}, so we focus now on the implicit regularizer
\begin{align}
  R(\Theta) = \sum_{c \in C} \sum_{i \in I} \hat{y}(c,i)^2 \label{eq:lci}
\end{align}
The general computational complexity is $\O(|C||I|)$.

In this section, we introduce the concept of a \separable{k} model.
We will provide an efficient implicit CD solver for any \separable{k} model.
In Section~\ref{sec:applications}, we show that many common models are \separable{k}, including matrix factorization, feature-based approaches such as factorization machines, but also higher-order tensor factorization such as PARAFAC or Tucker decomposition.
The iCD framework that we derive in this section is not limited to the models described above but can serve as a blueprint for other \separable{k} models as well.

\newtheorem{definition}{Definition}
\begin{definition}[\separable{k}]
  \label{definition:separable}
A model $\hat{y}(c,i)$ is called \underline{\separable{k}} iff the model can be rewritten as
\begin{align}
  \hat{y}(c,i) = \langle \bphi(c), \bpsi(i)\rangle = \sum_{f=1}^k \phi_f(c) \, \psi_f(i) \label{eq:separable}
\end{align}
with functions
\begin{align}
  \bphi : C \rightarrow \mathbb{R}^k, \quad \bpsi : I \rightarrow \mathbb{R}^k
\end{align}
where $\bphi$ is parameterized by $\Theta^C$ and $\bpsi$ is parameterized by $\Theta^I$ with $\Theta^C \cap \Theta^I = \emptyset$.
\end{definition}

\begin{lemma}
\label{lemma:decomposition}
The implicit regularizer of any \separable{k} model can be decomposed to:
\begin{align}
  R(\Theta) = \sum_{f=1}^k \sum_{f'=1}^k \underbrace{\sum_{c \in C} \phi_f(c)\,\phi_{f'}(c)}_{=:J_C(f,f')} \underbrace{\sum_{i \in I} \psi_f(i)\,\psi_{f'}(i)}_{=:J_I(f,f')} \label{eq:dloss}
\end{align}
\end{lemma}
\begin{proof}
The lemma follows from inserting the \separable{k} model (eq.~\ref{eq:separable}) into the implicit regularizer (eq.~\ref{eq:lci}) and rearranging the summations.
\begin{align*}
  R(\Theta) &= \sum_{c \in C} \sum_{i \in I} \sum_{f=1}^k \phi_f(c) \, \psi_f(i)  \sum_{f'=1}^k \phi_{f'}(c) \, \psi_{f'}(i) \\
		   &= \sum_{f=1}^k \sum_{f'=1}^k \left( \sum_{c \in C} \phi_f(c)\,\phi_{f'}(c) \right) \left(\sum_{i \in I} \phi_f(i)\,\phi_{f'}(i) \right)
\end{align*}
\end{proof}
This lemma is key to efficient learning algorithms from implicit data.
It shows that the context and item sides can be computed independently, which drops the computational complexity from $\O(|C|\,|I|)$ to $\O((|C| + |I|)\,k^2)$.
Next, we show how this can be used for gradient computation which is required for the update step in CD (see eq.~\ref{eq:newton}).
\begin{lemma}
\label{lemma:gradients}
The implicit regularizer gradients of any \separable{k} model with respect to any model parameter $\theta \in \Theta^C$ (or analogously $\theta \in \Theta^I$), can be simplified to
\begin{align}
  R'(\theta) = 2\,\sum_{f=1}^k \sum_{f'=1}^k J_I(f,f') \sum_{c \in C} \phi_f(c)\,\phi'_{f'}(c) \label{eq:dlossp}
\end{align}
\begin{align}
  R''(\theta) = 2\,\sum_{f=1}^k \sum_{f'=1}^k J_I(f,f') \sum_{c \in C} \left[ \phi_f(c)\,\phi''_{f'}(c) + \phi'_f(c)\,\phi'_{f'}(c) \right] \label{eq:dlosspp}
\end{align}
\end{lemma}
\begin{proof}
  The lemma follows from deriving eq.~(\ref{eq:dloss}).
\end{proof}
This lemma shows that computing $R'$ and $R''$ of any context parameter is independent of $|I|$.

\begin{algorithm}[t]
  \caption{Generic Implicit CD}
  \label{alg:cdgeneric}
  \begin{algorithmic}[1]
    \Procedure{iCD-Generic}{$S, C, I$}
      \State $\Theta \leftarrow \mathcal{N}(0, \sigma)$
      \Repeat
        \State Compute $\Phi$ and $\Psi$ if necessary
        \State Compute $J_I$  \label{line:cd_generic_context_begin}
        \For{$\theta \in \Theta^C$}
          \State Compute $L'(\theta|S), L''(\theta|S)$ 
            \State Compute $R'(\theta), R''(\theta)$ 
            \State $\theta \leftarrow \theta - \eta \frac{L'(\theta|S) + \alpha_0\,R'(\theta)}{L''(\theta|S) + \alpha_0\, R''(\theta)}$
            \State Update $\Phi$ if necessary
        \EndFor \label{line:cd_generic_context_end}
        \State Apply step \ref{line:cd_generic_context_begin} to \ref{line:cd_generic_context_end} to the items.
      \Until{converged}
      \State \Return $\Theta$
    \EndProcedure
  \end{algorithmic}
\end{algorithm}

From the analysis follows the recipe to derive an efficient iCD learning algorithm for a model $\hat{y}$.
First, rewrite the model as a dot product of $\bphi$ and $\bpsi$.
Second, construct the first and second derivative of $\bphi$ and $\bpsi$ with respect to any model parameter $\theta \in \Theta$.
These results allow to compute $R'(\theta)$ and  $R''(\theta)$ for any model parameter $\theta \in \Theta$ efficiently.
With these gradients for the expensive implicit regularizer, a Newton step can be applied.
Algorithm~\ref{alg:cdgeneric} shows a generic iCD algorithm using the ideas of this section.

Most models allow some further optimizations:
(i)~When the gradients of $\bphi$ or $\bpsi$ are sparse, some of the summands of eqs.~(\ref{eq:dlossp}, \ref{eq:dlosspp}) drop.
(ii)~The model parameters usually have some structure which can be used for traversing the model parameters more systematically.
We will show both of these steps in the next section for a variety of models.

\section{Applications}
\label{sec:applications}

In this section, we apply iCD to two classes of complex factorization models, namely feature-based factorization models and tensor factorization models.
We have chosen these two classes because they are very powerful and frequently used.
Moreover each of them has some interesting properties with respect to deriving iCD algorithms.
The provided algorithms can be directly applied to many common recommender system tasks.
This section also serves as a guide for deriving iCD algorithms in general.

\subsection{Matrix Factorization (MF)}

\begin{algorithm}[t]
  \caption{Implicit CD for MF}
  \label{alg:cdmf}
  \begin{algorithmic}[1]
    \Procedure{iCD-MF}{$S, C, I$}
      \State $W,H \leftarrow \mathcal{N}(0, \sigma)$
      \Repeat
        \For{$f^* \in \{1,\ldots,k\}$}
          \For{$f \in \{1,\ldots,k\}$} \label{line:cd_mf_context_begin}
            \State Compute $J_I(f^*, f)$
          \EndFor
          \For{$c^* \in C$}
            \State Compute $L'(w_{c^*,f^*}|S), L''(w_{c^*,f^*}|S)$
            \State Compute $R'(w_{c^*,f^*}), R''(w_{c^*,f^*})$
            \State $w_{c^*,f^*}\!\leftarrow\! w_{c^*,f^*}\!-\!\frac{L'(w_{c^*,f^*}|S) + \alpha R'(w_{c^*,f^*})}{L''(w_{c^*,f^*}|S) + \alpha R''(w_{c^*,f^*})}$
          \EndFor \label{line:cd_mf_context_end}
          \State Apply step \ref{line:cd_mf_context_begin} to \ref{line:cd_mf_context_end} to the items.
        \EndFor
      \Until{converged}
      \State \Return $W,H$
    \EndProcedure
  \end{algorithmic}
\end{algorithm}

\begin{figure}[t]
  \centering
  \includegraphics[width=0.23\textwidth]{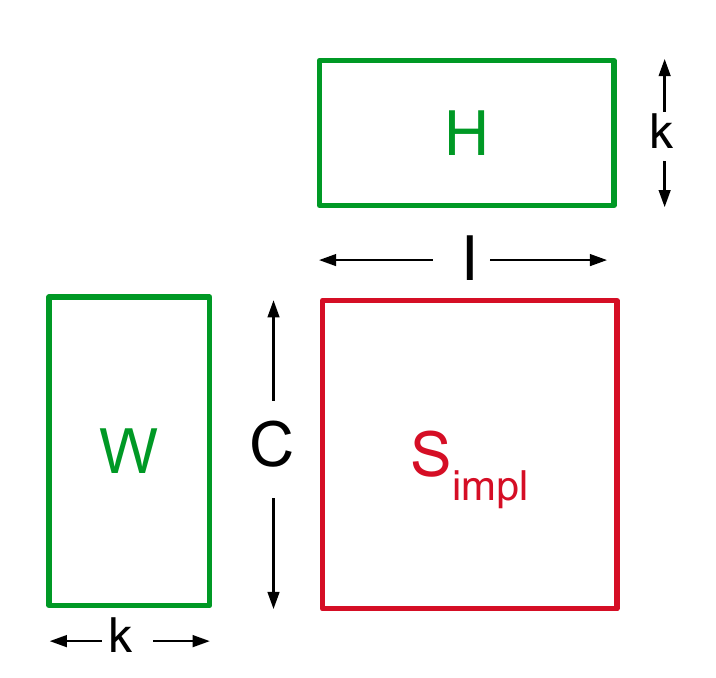}
  \caption{Matrix factorization:
   Each context $c$ is associated with an embedding $\bw_c$ and every item $i$ with an embedding $\bh_i$.
   The model parameters $W \in \mathbb{R}^{C \times k}, H \in \mathbb{R}^{I \times k}$ are learned to approximate the data $S_\text{impl}$ with the dot product of $\langle \bw_c, \bh_i\rangle$.
  \label{fig:mf}}
\end{figure}

We start by applying our framework to matrix factorization (see Figure~\ref{fig:mf}).
For MF, the scoring function is
\begin{align}
  \hat{y}(c,i) := \langle \textbf{w}_c, \textbf{h}_i \rangle = \sum_{f=1}^k w_{c,f}\, h_{i,f}
\end{align}
with model parameters $\Theta = \{W, H\}$ where $W \in \mathbb{R}^{C \times k}$ and $H \in \mathbb{R}^{I \times k}$.

A MF model is trivially \separable{k} with
\begin{align}
  \phi_f(c) = w_{c,f},\quad 
  \psi_f(i) = h_{i,f}.
\end{align}
Furthermore, the gradients are sparse
\begin{align}
  \frac{\partial \phi_f(c)}{\partial w_{c^*,f^*}} =
  \begin{cases}
    1, &\text{if } c = c^* \wedge f = f^* \\
    0, &\text{else}
  \end{cases}
\end{align}
and all second derivatives are 0.
Thus, the regularizer derivatives simplify to
\begin{align}
  R'(w_{c^*,f^*}) &= 2 \, \sum_{f=1}^k J_I(f,f^*)\, w_{c^*,f} \\
  R''(w_{c^*,f^*}) &= 2 \, J_I(f^*,f^*)
\end{align}
The derivation is symmetric for the item side.

As MF associates each model parameter with an embedding dimension $f$, we can traverse the parameters one dimension at a time.
A full step $\eta=1$ can be taken because MF is bilinear.
Algorithm \ref{alg:cdmf} shows the full procedure.

The computation of $J_I(f^*, \cdot)$ is trivially in $\O(|I|\,k)$.
Gradient computation of the implicit regularizer is $\O(k)$ per parameter and for the explicit part $\O(|S|)$ for all parameters.
Overall, the algorithm has a complexity of $\O((|I|+|C|)\,k^2 + |S|\,k)$ per iteration.

\subsection{Feature-Based Factorization Models}
\label{sec:feature_based_factorization_models}

One of the most powerful extension of MF is feature based modeling for the context and item.
Feature-based factorization models are strictly more powerful than MF and have shown large improvements in many applications (e.g. \cite{GantnerDrumondFreudenthalerEtAl2010,Rendle:tist2012}).
For instance, the cold-start problem is commonly solved by replacing or complementing user and item ids with user and item attributes \cite{GantnerDrumondFreudenthalerEtAl2010}.
Another example is context-aware recommendation, where the context is represented by several variables, e.g. location or time in addition to the user id.
Also sequential models can be represented by feature based modeling \cite{kanagal:vldb12}.

Learning general feature-based models on implicit feedback was restricted to BPR so far.
This is the first work that provides an implicit CD algorithm for this important model class.

\begin{figure}[t]
  \centering
  \includegraphics[width=0.27\textwidth]{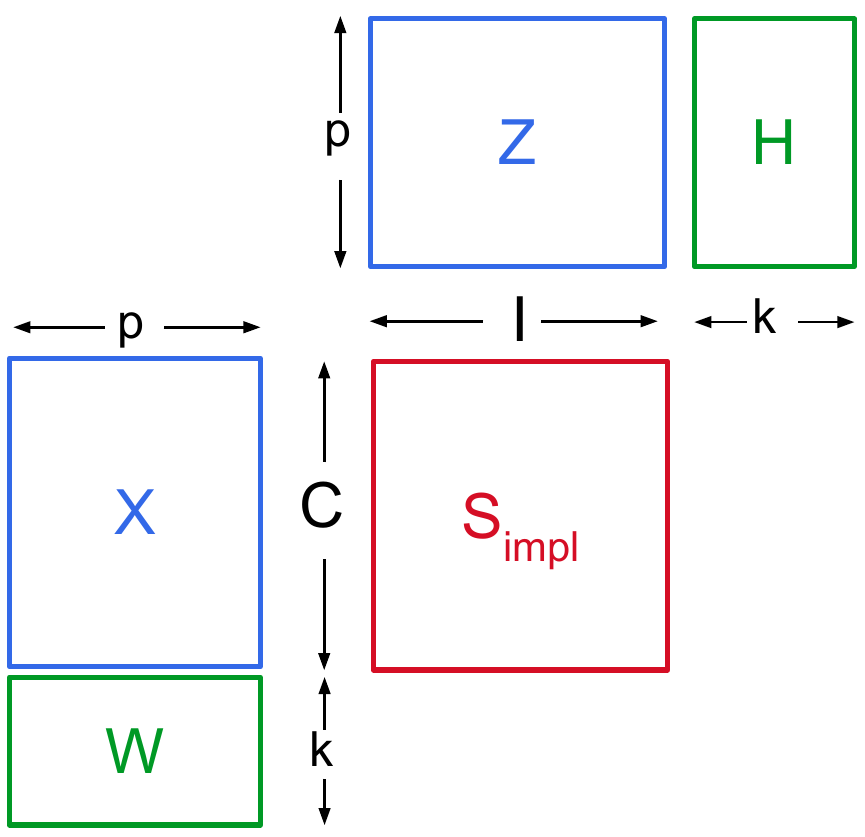}
  \caption{Matrix factorization with side information:
   In addition to the data $S_\text{impl}$, a feature vector $\bx_c \in \mathbb{R}^p$ is given for every context $c \in C$ and a feature vector $\bz_i \in \mathbb{R}^p$ for each item $i \in I$.
   Each of the context features $l \in \{1,\ldots,p\}$ is assigned a $k$-dimensional embedding vector $\bw_l \in \mathbb{R}^k$ and similarly $\bh_l \in \mathbb{R}^k$ for each item feature.
   The model parameters $W \in \mathbb{R}^{p \times k}, H \in \mathbb{R}^{p \times k}$ are learned to approximate the data $S_\text{impl}$ with a $X\,W\,(Z\,H)^t$.
  \label{fig:mf_side}}
\end{figure}

To formalize the problem, assume each $c \in C$ is represented by a feature vector $\bx_c \in \mathbb{R}^p$ and each $i \in I$ by a feature vector $\bz_i \in \mathbb{R}^p$.
See Figure~\ref{fig:mf_side} for an illustration. 

\subsubsection{MF with Side Information (MFSI)}

We start with a feature based extension of matrix factorization similar to \cite{GantnerDrumondFreudenthalerEtAl2010}:
\begin{align}
  \hat{y}(c,i) = \bx_c\,W\,(\bz_i\,H)^t = \sum_{f=1}^k \left(\sum_{l=1}^p x_{c,l} \, w_{l, f}\right) \left(\sum_{l=1}^p z_{i,l} \, h_{l, f}\right)
\end{align}
with $\Theta = \{W,H\}$.
MFSI is \separable{k} using
\begin{align}
  \phi_f(c) = \sum_{l=1}^p x_{c,l} \, w_{l, f},\quad \psi_f(i) = \sum_{l=1}^p z_{i,l} \, h_{l, f}
\end{align}
and the gradients are sparse
\begin{align}
  \frac{\partial \phi_f(c)}{\partial w_{l^*,f^*}} =
  \begin{cases}
    x_{c,l^*}, &\text{if } f = f^* \\
    0, &\text{else}
  \end{cases}
\end{align}
Due to sparse gradients of $\bphi$ and $\bpsi$, the first and second regularizer derivatives simplify to:
\begin{align}
  R'(w_{l^*,f^*}) &= 2 \, \sum_{f=1}^k J_I(f,f^*)\, \sum_{c \in C} x_{c,l^*}\, \phi_f(c) \\
  R''(w_{l^*,f^*}) &= 2 \, J_I(f^*,f^*) \sum_{c \in C} x_{c,l^*}^2
\end{align}
Note that the sums over the context variable depend only on context where $x_{c,l^*} \ne 0$, so with a sparse iterator, the computation is $\O(k\,N_Z(X))$ for optimizing all of the context variables in a given embedding layer $f^*$.

This computation assumes that $\Phi$ and $\Psi$ are given.
Obviously, while optimizing $W$, $\Psi$ does not change and while optimizing $H$, $\Phi$ does not change.
However, while optimizing $W$, $\Phi$ changes but can be kept in sync with changes in $W$ by updating:
\begin{align}
  \phi_{f^*}(c) \leftarrow \phi_{f^*}(c) + x_{c,l^*} (w^{\text{new}}_{l^*,f^*} - w^{\text{old}}_{l^*,f^*})
\end{align}  
The item side can be derived analogously.
The total runtime of Algorithm~\ref{alg:cdsi} for one epoch over all variables is $\O(k^2\,(N_Z(X) + N_Z(Z)))$ for the implicit regularizer.

\begin{algorithm}[t]
  \caption{Implicit CD for MF with Side Information}
  \label{alg:cdsi}
  \begin{algorithmic}[1]
    \Procedure{iCD-MFSide}{$S, C, I$}
      \State $W,H \leftarrow \mathcal{N}(0, \sigma)$
      \Repeat
        \State Compute $\Phi$ and $\Psi$
        \For{$f^* \in \{1,\ldots,k\}$} 
          \For{$f \in \{1,\ldots,k\}$} \label{line:cd_mfsi_context_begin} 
            \State Compute $J_I(f^*, f)$
          \EndFor
          \For{$l^* \in \{1,\ldots,p\}$}
            \State Compute $L'(w_{l^*,f^*}|S), L''(w_{l^*,f^*}|S)$ 
            \State Compute $R'(w_{l^*,f^*}), R''(w_{l^*,f^*})$ 
            \State $w_{l^*,f^*}\!\leftarrow\!w_{l^*,f^*}\!-\!\frac{L'(w_{l^*,f^*}|S) + \alpha R'(w_{l^*,f^*})}{L''(w_{l^*,f^*}|S) + \alpha R''(w_{l^*,f^*})}$
            \State Update $\Phi$
          \EndFor \label{line:cd_mfsi_context_end}
          \State Apply step \ref{line:cd_mfsi_context_begin} to \ref{line:cd_mfsi_context_end} to the items.
        \EndFor
      \Until{converged}
    \EndProcedure
  \end{algorithmic}
\end{algorithm}

\subsubsection{Factorization Machines}

The Factorization Machine (FM) model \cite{Rendle:tist2012} is a more complex factorized model that includes biases and interactions between all variables.
In general, the FM for a feature vector $\bx \in \mathbb{R}^p$ is defined as
\begin{align}
  \hat{y}(\bx) = b + \sum_{l=1}^p x_l\,\tilde{w}_l + \sum_{l=1}\sum_{l'>l} \langle \bw_{l}, \bw_{l'} \rangle x_l x_l'
\end{align}
where $b$ is a global bias parameter, $\tilde{\bw}$ are feature biases and $W$ are the embeddings.
In our case, for a context-item pair $(c,i)$, we set the input feature vector $\bx$ as the concatenation of the context and item feature vectors: $\bx := (\bx_c, \bz_i)$.

The FM model is \separable{(k+2)}, with
\begin{align}
  \phi_f(c) &= \sum_{l=1}^p x_{c,l} \, w_{l, f},\quad  
  \psi_f(i) = \sum_{l=1}^p z_{i,l} \, h_{l, f},\\
  \phi_{k+1}(c) &= b + \sum_{l=1}^p x_{c,l} \, \tilde{w}_{l} + \sum_{l=1}\sum_{l'>l} \langle \bw_{l}, \bw_{l'} \rangle x_{c,l} x_{c,l'},\\
  \psi_{k+1}(i) &= 1,\\
  \phi_{k+2}(c) &= 1,\\
  \psi_{k+2}(c) &= \sum_{l=1}^p z_{i,l} \, \tilde{h}_{l} + \sum_{l=1}\sum_{l'>l} \langle \bh_{l}, \bh_{l'} \rangle z_{i,l} z_{i,l'} .
\end{align}
where for the context, $\bphi$ is parameterized by $\tilde{\bw} \in \mathbb{R}^p$ for the linear part and $W \in \mathbb{R}^{p \times k}$ for the factors.
And analogously for items, $\bpsi$ is parameterized by $\tilde{\bh} \in \mathbb{R}^p$ for the linear part and $H \in \mathbb{R}^{p \times k}$ for the factors.

The gradients are sparse:
\begin{align}
  \frac{\partial \phi_{f}(c)}{\partial \tilde{w}_{l^*}} &= \begin{cases}
    x_{c,l^*}, &\text{if } f = k+1 \\
    0,&\text{else}
  \end{cases}\\
  \frac{\partial \phi_f(c)}{\partial w_{l^*,f^*}} &= \begin{cases}
    x_{c,l^*}, &\text{if } f = f^* \\
    x_{c,l^*} (\phi_{f^*}(c) - x_{c,l^*} w_{l^*,f^*}), &\text{if } f = k+1\\
    0, &\text{else}
  \end{cases}
\end{align}
Similar to MFSI, due to the sparsity in gradients, one of the nested loops drops for the first regularizer derivative $R'$ and both nested sums drop for the second regularizer derivative $R''$.
Consequently, the flow and runtime analysis for FM is the same as for MFSI.

\subsection{Tensor Factorization}

\begin{figure}[t]
  \centering
  \includegraphics[width=0.4\textwidth]{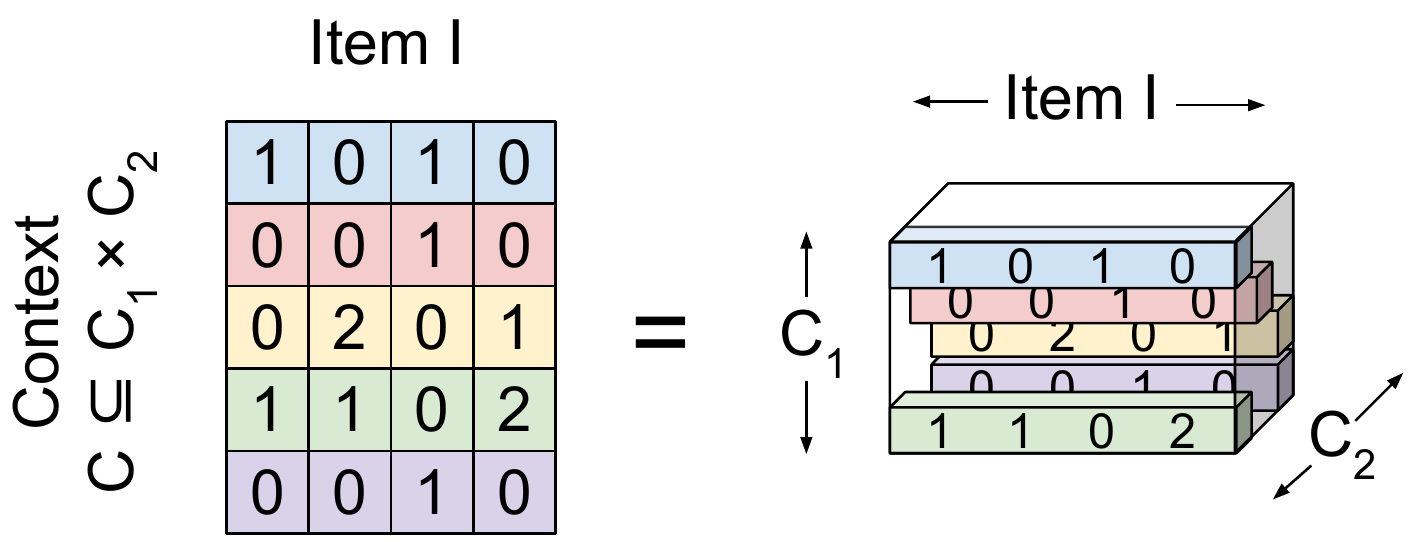}
  \caption{Example for tensor factorization data: two categorical variables on the context side, $C \subseteq C_1 \times C_2$. This data can be interpreted as a 3-mode tensor with missing values.
  \label{fig:tf_data}}
\end{figure}

Tensor factorization generalizes matrix factorization and deals with problems that involve more than two categorical variables.
For instance, in personalized recommendation of tags for bookmarks~\cite{Symeonidis:tkde10}, the \emph{context} consists of two variables, the user $C_1$ and the bookmark $C_2$, and the item $I$ corresponds to the tag.
For personalized web search~\cite{Sun:www2005}, the context consists of the user $C_1$ and the query $C_2$ and the item $I$ to the web page.
The data can be seen as a three mode tensor over $C_1$, $C_2$ and $I$.
Figure~\ref{fig:tf_data} shows an example of how observations over context $C \subseteq C_1 \times C_2$ and items $I$ translate to a tensor.
A tensor factorization model tries to approximate the tensor with a low rank decomposition (see Figure~\ref{fig:tf}).
Although tensor factorization models are multilinear, we show that they fit well into our framework.

Additionally, we want to highlight, that existing tensor factorization learning algorithms~\cite{Sun:www2005,Symeonidis:tkde10,PilaszyZibriczkyTikk2010} require that the tensor data is dense, i.e., the empty parts in the tensor in Figure~\ref{fig:tf_data} are filled with zeros.
This would imply that context combinations that never have been observed, are used for training as well, i.e., $C=C_1 \times C_2$.
In some applications, this might not make sense, for instance if $C_1$ encodes a device type and $C_2$ encodes an operating system version.
Our iCD framework works for both sparse and dense context.
We will point out the differences when necessary.

\begin{figure}[t]
  \centering
  \includegraphics[width=0.27\textwidth]{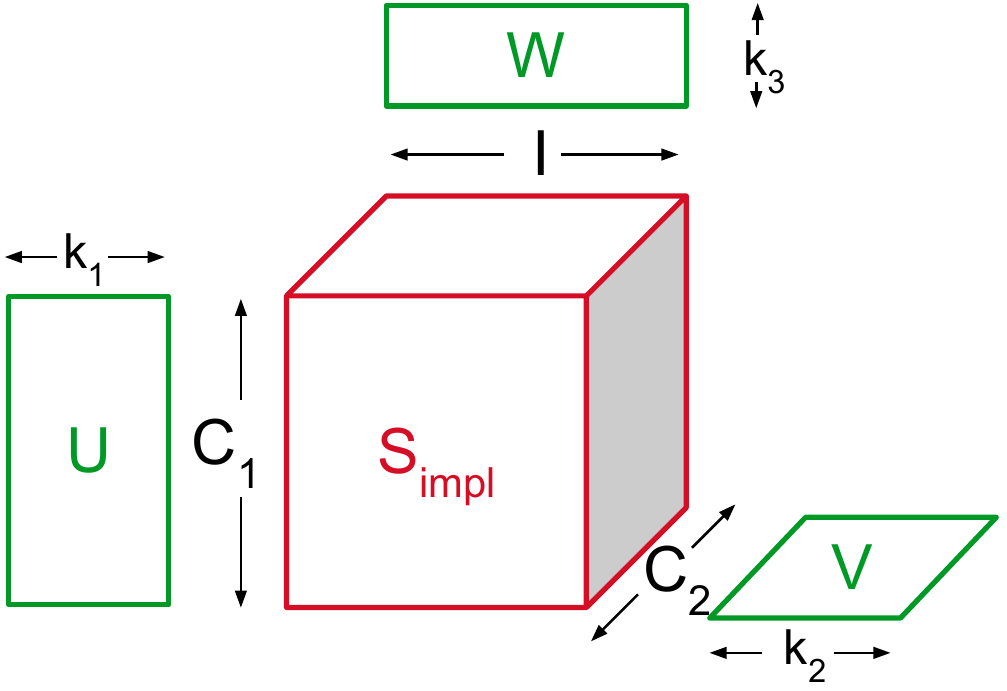}
  \caption{A tensor factorization model, decomposes a given tensor into one matrix per mode, here $U \in \mathbb{R}^{C_1 \times k_1}$ for $C_1$, $V \in \mathbb{R}^{C_2 \times k_2}$ for $C_2$ and $W \in \mathbb{R}^{C_3 \times k_3}$ for $I$.
  \label{fig:tf}}
\end{figure}

\subsubsection{Parallel Factor Analysis (PARAFAC)}

We first discuss the Parallel Factor Analysis (PARAFAC)~\cite{harshman:1970} model which is a 3-mode extension of matrix factorization.
\begin{align}
  \hat{y}(c_1, c_2, i) := \sum_{f=1}^k u_{c_1,f}\,v_{c_2,f}\, w_{i,f}
\end{align}
with $\Theta = \{U, V, W\}$ where $U \in \mathbb{R}^{C_1 \times k}$, $V \in \mathbb{R}^{C_2 \times k}$ and $W \in \mathbb{R}^{I \times k}$.
PARAFAC is \separable{k} with
\begin{align}
  \phi_f(c_1, c_2) = u_{c_1,f}\,v_{c_2, f},\quad 
  \psi_f(i) = w_{i,f}
\end{align}
Again, gradients are sparse:
\begin{align}
  \frac{\partial \phi_f(c_1, c_2)}{\partial u_{c_1^*,f^*}} =
  \begin{cases}
    v_{c_2,f}, &\text{if } c_1 = c_1^* \wedge f = f^* \\
    0, &\text{else}
  \end{cases}
\end{align}
and the loss derivatives simplify to
\begin{align}
  R'(u_{c_1^*,f^*}) &= 2 \, \sum_{f=1}^k J_I(f,f^*)\,u_{c^*,f} \sum_{c_2: (c^*_1, c_2) \in C} v_{c_2,f}\, v_{c_2,f^*} \label{eq:parafaclp}\\
  R''(u_{c_1^*,f^*}) &= 2 \, J_I(f^*,f^*)\sum_{c_2: (c^*_1, c_2) \in C} v_{c_2,f^*}\, v_{c_2,f^*} \label{eq:parafaclpp}
\end{align}
The item side is equivalent to matrix factorization.

If the context is dense and includes all possible combinations of context variables, i.e., if $C = C_1 \times C_2$, then the computation of $J_C(f,f')$, can be decomposed to:
\begin{align}
  J_C(f,f') = \underbrace{\sum_{c_1 \in C} u_{c_1,f}\,u_{c_1,f'}}_{=:J_{C_1}(f,f')} \underbrace{\sum_{c_2 \in C} v_{c_2,f}\,v_{c_2,f'}}_{=:J_{C_2}(f,f')}
\end{align}
This means, the computation is in $\O(|C_1|+|C_2|)$ instead of $\O(|C_1|\,|C_2|)$.
On the other hand if $C$ is sparse and contains only the subset of the observed context combinations, i.e., $C \subset C_1 \times C_2$, then there is no need for decomposing this sum.
The same applies to the loss derivatives of eqs. (\ref{eq:parafaclp},\ref{eq:parafaclpp}):
Again, if all possible context is modeled, then $\{c_2: (c^*_1, c_2) \in C\} = C_2$ and thus $J_{C_2}(f,f')$ can replace the sum over $C_2$.

The overall runtime for PARAFAC's implicit regularizer is $\O((|C| + |I|)\,k^2)$ for sparse context and $\O((|C_1| + |C_2| + |I|)\,k^2)$ for dense context.
The traversal over model parameters can be arranged as in the MF algorithm.

\subsubsection{Tucker Decomposition}

Tucker Decomposition (TD)~\cite{tucker:1966} is a generalization of PARAFAC which computes all interactions between the factor matrices.
The strength of each interaction is given by a core tensor $B$.
For our running example with two context variables $c_1,c_2$ and one item variable $i$, TD is defined as
\begin{align}
  \hat{y}(c_1, c_2, i) = \sum_{f_1=1}^{k_1} \sum_{f_2=1}^{k_2} \sum_{f_3=1}^{k_3} b_{f_1,f_2,f_3} u_{c_1,f_1}\,v_{c_2,f_2}\,w_{i,f_3}
\end{align}
with $\Theta = \{B, U, V, W\}$ where $B \in \mathbb{R}^{k_1 \times k_2 \times k_3}$ is the core tensor and $U \in \mathbb{R}^{|C_1| \times k_1}$, $V \in \mathbb{R}^{|C_2| \times k_2}$ and $W \in \mathbb{R}^{|I| \times k_3}$.
TD is much more computationally expensive than PARAFAC, requiring $\O(k_1\,k_2\,k_3)$ operations just for evaluating the model on one data point.

Even though Tucker decomposition contains nested sums, it is \separable{k_3} with
\begin{align*}
  \phi_f(c_1, c_2) = \sum_{f_1=1}^{k_1} \sum_{f_2=1}^{k_2} b_{f_1,f_2,f} u_{c_1,f_1}\,v_{c_2, f_2},\quad 
  \psi_f(i) = w_{i,f}
\end{align*}
The derivatives of these functions are:
\begin{align}
  \frac{\partial \phi_f(c_1,c_2)}{\partial u_{c_1,f_1^*}} &= \begin{cases}
    \sum_{f_2=1}^{k_2} b_{f_1^*,f_2,f}\,v_{c_2,f_2}, &\text{if } c_1 = c_1^*\\
    0, &\text{else}
  \end{cases}\\
  \frac{\partial \phi_f(c_1,c_2)}{\partial v_{c_2,f_2^*}} &= \begin{cases}
    \sum_{f_1=1}^{k_1} b_{f_1,f_2^*,f}\,u_{c_1,f_1}, &\text{if } c_2 = c_2^*\\
    0, &\text{else}
  \end{cases}\\
  \frac{\partial \phi_f(c_1,c_2)}{\partial b_{f_1^*,f_2^*,f_3^*}} &= \begin{cases}
    u_{c_1,f^*_1}\,v_{c_2, f^*_2}, &\text{if } f = f_3^* \\
    0, &\text{else}
  \end{cases}\\
  \frac{\partial \psi_f(i)}{\partial w_{i^*, f_3^*}} &=\begin{cases}
    1, &\text{if } f = f_3^* \wedge i = i^*\\
    0, &\text{else}
  \end{cases}
\end{align}
Unlike all the other models we have presented so far, the gradients for $\bphi$ are non-zero for any factor index $f \in \{1,\ldots, k\}$.
Consequently, the nested loops over factors of the loss gradient (eq.~\ref{eq:dlossp}) cannot be improved further.
However, for $\bpsi$, which is sparse, the same optimization as in the other models can be applied.

Like for PARAFAC, if $C$ is dense, i.e., $C=C_1 \times C_2$, we can precompute intermediate matrices for $C_1$ and $C_2$ and the computation of $J_C(f,f')$ simplifies to
\begin{align*}
  \sum_{f_1=1}^{k_1} \sum_{f'_1=1}^{k_1} \sum_{f_2=1}^{k_2} \sum_{f'_2=1}^{k_2} b_{f_1,f_2,f} \,b_{f'_1,f'_2,f'} J_{C_1}(f_1,f'_1) J_{C_2}(f_2,f'_2)
\end{align*}
If $C$ is sparse, there is no need for this optimization and we can use a straightforward computation of $J_C$.
The overall runtime complexities are $\O(k_1^2\,k_2^2\,k_3^2\,(|C_1|+|C_2|+|I|))$ for dense context and $\O(k_1^2\,k_2^2\,k_3^2\,(|C|+|I|))$ for sparse context.

\section{Experiments}
\label{sec:experiments}

\begin{figure*}
\centering
\begin{subfigure}[b]{0.49\textwidth}
\includegraphics[width=0.49\textwidth]{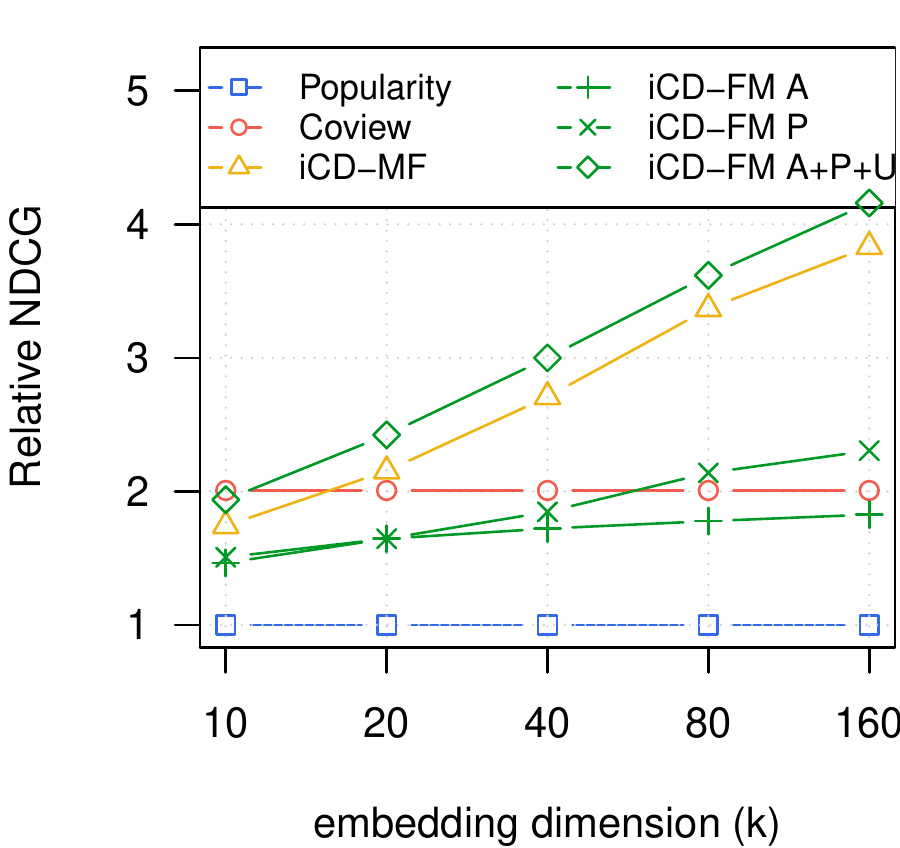}
\includegraphics[width=0.49\textwidth]{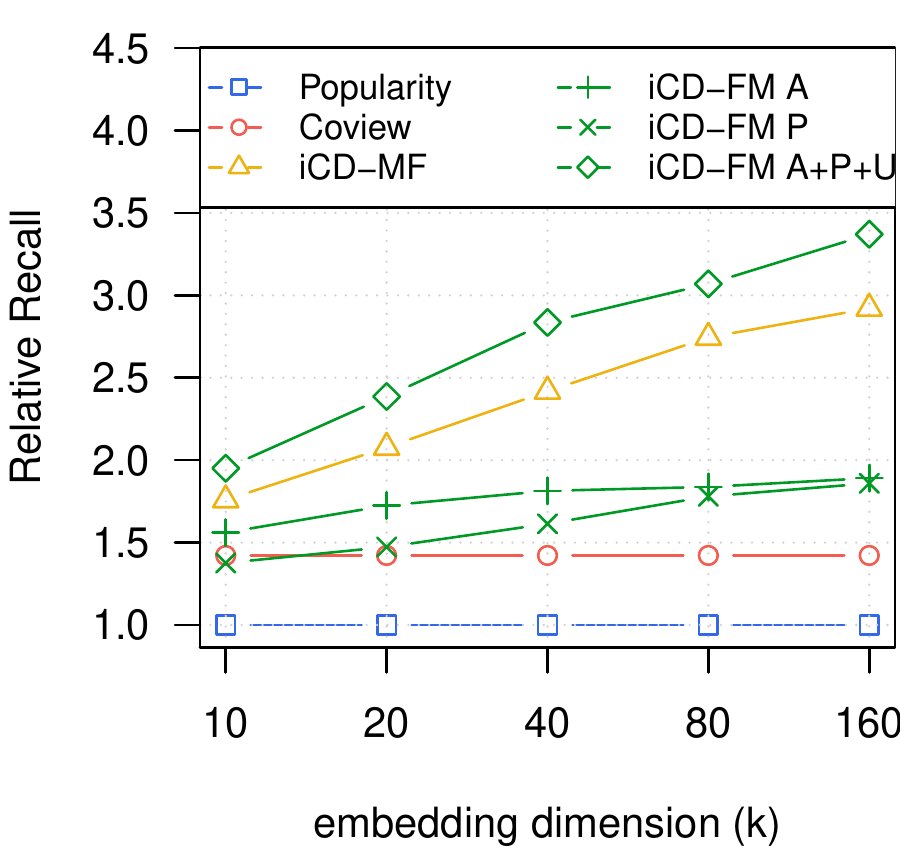}
\caption{Offline Recommendation\label{fig:offline_recommendation}}
\end{subfigure}
\begin{subfigure}[b]{0.49\textwidth}
\includegraphics[width=0.49\textwidth]{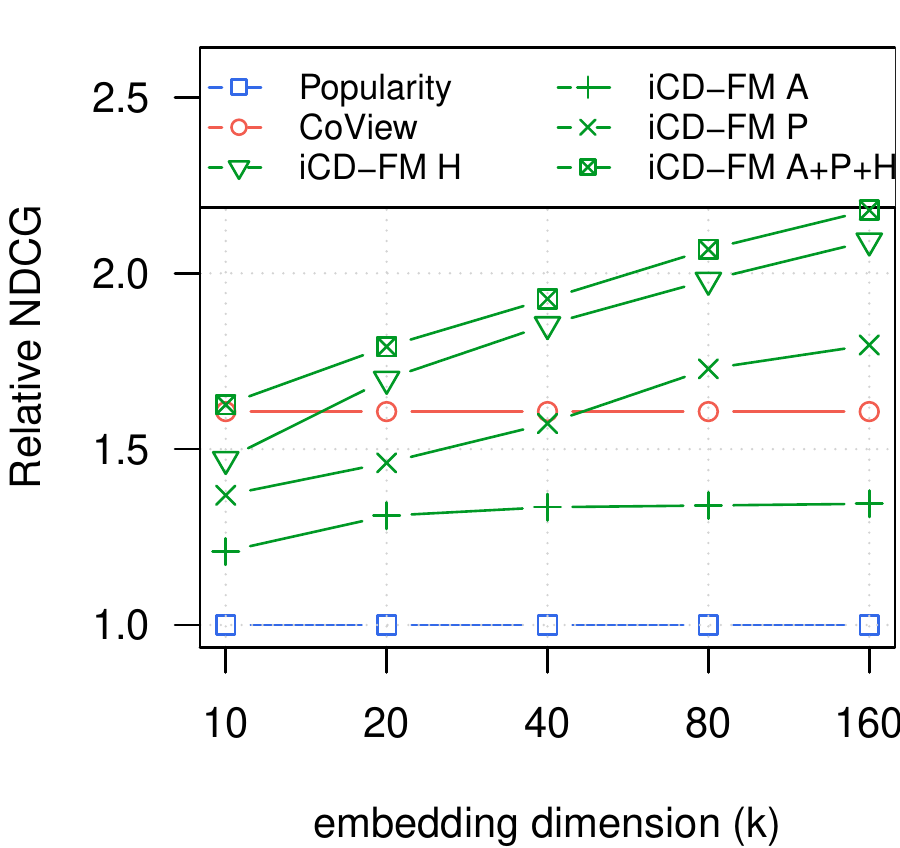}
\includegraphics[width=0.49\textwidth]{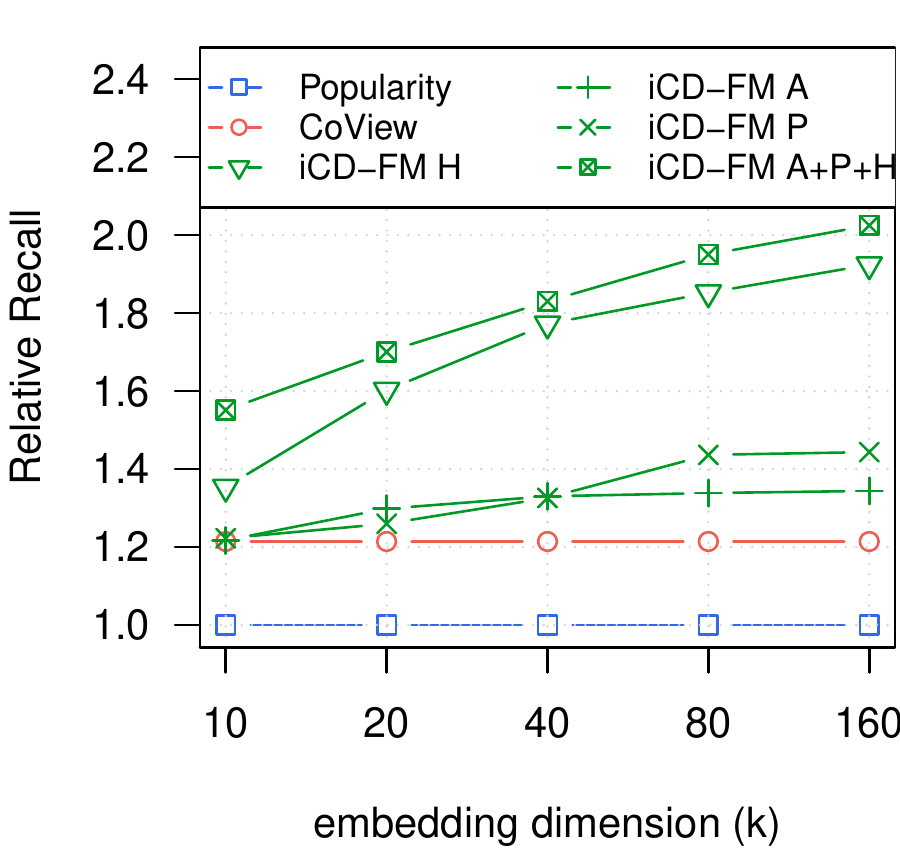}
\caption{Instant Recommendation\label{fig:online_recommendation}}
\end{subfigure}
  \caption{Different variants of context features are used in the {\tt iCD-FM} models: \emph{A} = gender, age, country and device, \emph{P} = the previously watched video, \emph{H} = all videos watched so far, \emph{U} = user id.\label{fig:quality}}
\end{figure*}

The main objective of the experiments is to illustrate the generality of the iCD
framework. We show how iCD can be applied to a variety of recommender problems
that cannot be solved with MF alone.  For MF models, efficient coordinate
descent algorithms (CD) have been previously proposed~\cite{Hu:icdm2008} and its
performance compared against gradient descent algorithms such as
BPR~\cite{rendle:uai09}. Both approaches are considered state-of-the-art and
while CD outperforms BPR on certain
datasets~\cite{NingKarypis2011,Zhao14,sedhain2016effectiveness,volkovs2015effective,zhao2015improving},
BPR has been shown to work better on
others~\cite{HeM16,ShiKaratzoglouBaltrunasEtAl2012,ShiKaratzoglouBaltrunasEtAl2012a,NingKarypis2011}.
The purpose of our experiments is not to compare BPR and
CD on yet another dataset, but rather to demonstrate the versatility of the
iCD framework and illustrate how it can serve as a building block for future
research on complex recommender models. As with MF, it is likely that both iCD
and BPR will show strengths in different applications.

\subsection{Experimental Setup}
We evaluate on a dataset of $200,000$ users interacting with YouTube.
Our subset contains $|I| = 68,000$ videos.
The dataset also contains side information about age, country, gender and device info.
We apply iCD to three popular recommendation problems -- Cold-Start, Offline
Recommendation, and Instant Recommendation (see Section~\ref{sec:results}).
We compare the following algorithms:
\begin{itemize}
  \item {\tt Popularity}: a static recommender that returns the most popular videos.
  \item {\tt Coview}: returns based on the previously watched video, the most commonly chosen next video.
  \item {\tt iCD-MF}: user-item matrix factorization using iCD for optimization, similar to~\cite{Hu:icdm2008}.
  \item {\tt iCD-FM}: a factorization machine with varying features for the context
  (Section~\ref{sec:feature_based_factorization_models}). We report results for different feature choices.
\end{itemize}
We measure the recall and NDCG for the top 100 returned videos.
Note that we report relative improvements over the {\tt Popularity} recommender.
All hyperparameters are tuned on a separate tuning holdout set.

\subsection{Results}
\label{sec:results}

\subsubsection{Cold-Start Recommendation}

In the Cold-Start recommendation~\cite{GantnerDrumondFreudenthalerEtAl2010}
scenario, we assume that a user interacts with the recommender system for the first time.
To simulate this scenario, we select a random subset of users and hold out all their events for evaluation purposes; we train on the remaining users.

The common approach for dealing with cold-start is to represent a user by side information~\cite{GantnerDrumondFreudenthalerEtAl2010}.
Here, we use the feature-based FM model ({\tt iCD-FM}) with the user's age,
gender, country and device info as context features.
Figure~\ref{fig:cold_start} shows that attribute-aware FM achieves a 2x improvement over the baselines.
As expected, neither {\tt MF} nor {\tt Coview} can do any better than most-popular recommendation.

\subsubsection{Offline Recommendation}

In the {\em Offline Recommendation} scenario, we hold out the last feedback of
each user and use all the previous feedback for training. This is the most
commonly used protocol to evaluate the performance of a recommender algorithm.
We experiment with multiple FM models: (1) {\tt iCD-FM A:} an FM with user attributes,
(2) {\tt iCD-FM P:} a sequential FM that only uses the previously watched video
(similar to FPMC~\cite{rendle:www10} or Coview) and (3) {\tt iCD-FM A+P+U:} an FM
that uses all signals: attributes, previously watched video and user id
(similar to FPMC~\cite{rendle:www10} with user attributes).
As shown in Figure~\ref{fig:offline_recommendation}, the complex FM model with all
features achieves the best quality, illustrating the flexibility of feature
engineering with iCD.

\subsubsection{Instant Recommendation}

In large-scale industrial applications, online training is often not feasible due to complex serving stacks.
Commonly, models are periodically trained offline (e.g., every day or week) and applied on a stream of user interactions.
When the model is queried to generate recommendations for a user, all
feedback until the current time is taken into account for prediction.
We simulate this setting by choosing a global cutoff time where all the events
before the cutoff are used for training and all the remaining ones for
evaluation.

In such settings, models relying on user ids, such as MF, cannot capture recent
feedback. Instead, describing a user by the sequence of previously watched
videos allows for instant personalization.
Such a model can be configured using a feature-based FM
model (Section~\ref{sec:feature_based_factorization_models}) and
we experiment with four configurations (1) {\tt iCD-FM A:} FM
using user attributes, (2) {\tt iCD-FM P:} a sequential FM based on the
previously watched video, (3) {\tt iCD-FM H:} a FM based on all previously
watched videos, (4) {\tt iCD-FM A+P+H:} an FM combining all signals. As
expected, the complex FM model with all features achieves the best quality.
Again, we would like to note the generality of the iCD framework, which enables
flexible feature engineering.

\subsection{Computational Costs}

As stated in Section~\ref{sec:implicit}, any conventional CD solver, e.g.~\cite{Rendle:tist2012}, could solve the implicit feedback problem. 
Now, we substantiate that this is infeasible because of the large number of implicit examples.
Figure~\ref{fig:comp_costs} compares the computational cost for learning an FM with a conventional CD to the costs of iCD on our dataset with 70k items.
We use three different context features from Figure~\ref{fig:quality}.
The plot shows relative costs to {\tt iCD-FM P}.
For all three context choices, conventional CD shows four orders of
magnitude higher compuational costs than iCD.
The empirical measured runtime for iCD was in the order of minutes; consequently, CD's four order of magnitude increase in runtime translates to weeks of training for each iteration.
Clearly, using a conventional CD solver to optimize the implicit loss directly is infeasible.

\begin{figure}[t]
\centering
\includegraphics[width=0.37\textwidth]{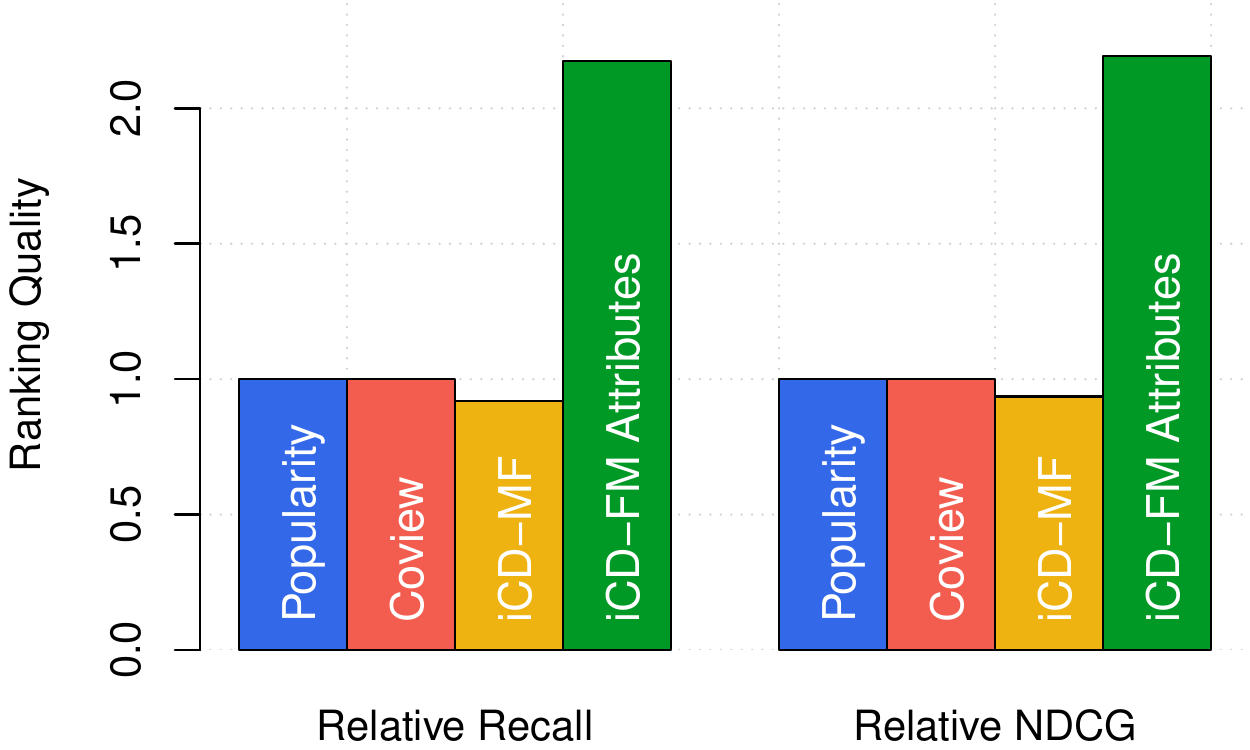}
  \caption{Cold Start Recommendation\label{fig:cold_start}}
\end{figure}

\begin{figure}[t]
\centering
\includegraphics[width=0.37\textwidth]{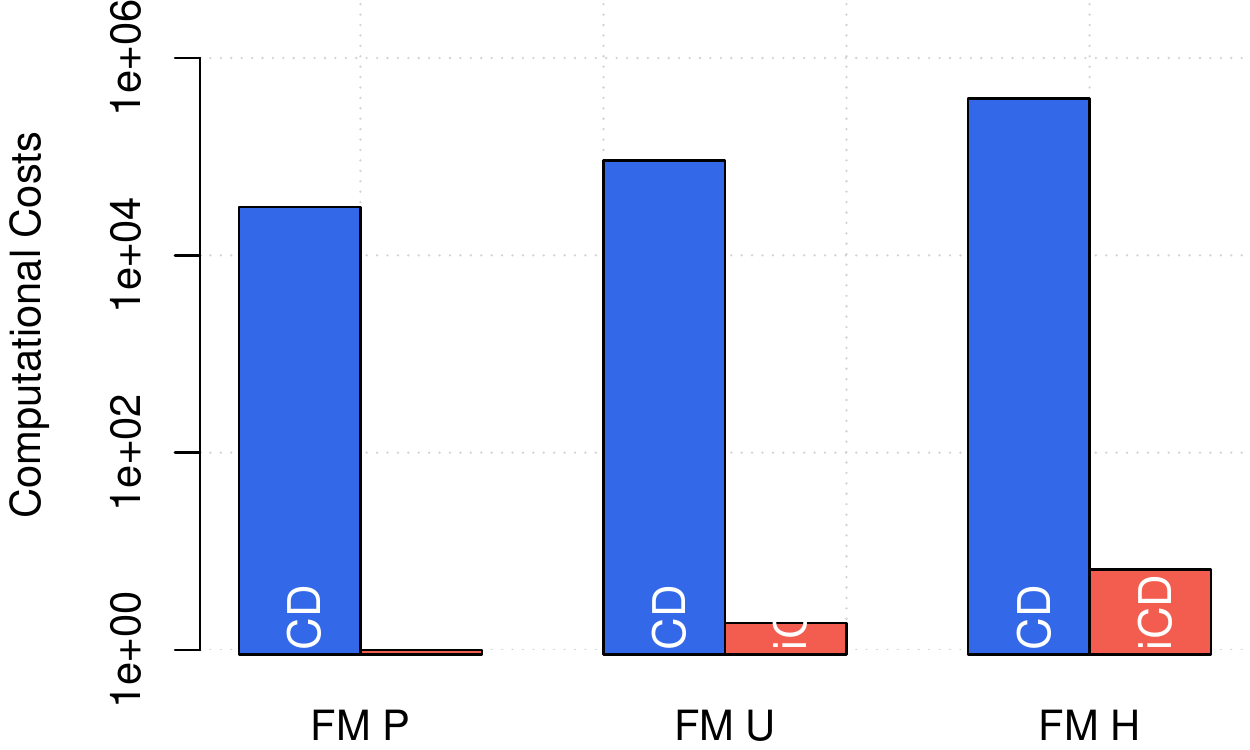}
  \caption{Training costs in log-scale of conventional CD (left, blue) versus iCD (right, red) on our implicit dataset.\label{fig:comp_costs}}
\end{figure}

\section{Conclusion}

In this work, we have presented a general, efficient framework for learning recommender system models from implicit feedback.
First, we have shown that learning from implicit feedback can be reformulated as optimizing a cheap explicit loss and an expensive implicit regularizer.
Then we have introduced the concept of \separable{k} models.
We have shown that the implicit regularizer of any \separable{k} model can be computed efficiently without iterating over all context-item pairs.
Finally, we have shown that many popular recommender models are \separable{k}, including matrix factorization, factorization machines and tensor factorization.
Moreover, we have provided efficient learning algorithms for these models based on our framework.
Our framework is not limited to the models discussed in the paper but designed to serve as a general blueprint for deriving learning algorithms for recommender systems.

\balance

\bibliographystyle{abbrv}
\bibliography{paper}
\end{document}